%% file: ijcai22.tex
\documentclass{article}
\usepackage{ijcai22}

\usepackage{times}
\usepackage{soul}
\usepackage{url}
\usepackage[hidelinks]{hyperref}
\usepackage[utf8]{inputenc}
\usepackage[small]{caption}
\usepackage{graphicx}
\usepackage{amsmath}
\usepackage{amsthm}
\usepackage{booktabs}
\usepackage{algorithm}
\usepackage{algorithmic}
\usepackage{amsmath}
\usepackage{hyperref} 
\hypersetup{
    colorlinks=true,
    citecolor=blue, 
    linkcolor=blue,
    filecolor=magenta,      
    urlcolor=blue,
linktocpage}
\usepackage{amsthm} 
\usepackage{mathtools} 
\usepackage{amsfonts}       

\usepackage{graphicx}
\usepackage{float}
\usepackage{subcaption}

\usepackage{extarrows}
\usepackage{thmtools}

\usepackage[toc,page,header]{appendix}
\usepackage{minitoc}

\newtheorem{theorem}{Theorem}
\newtheorem{corollary}{Corollary}
\newtheorem{lemma}{Lemma}
\newtheorem{proposition}{Proposition}
\newtheorem{definition}{Definition}

\newtheorem{example}{Example}

\DeclareMathOperator*{\argmax}{arg\,max}

\newcommand{\bfp}{\mathbf{p}}
\newcommand{\bfq}{\mathbf{q}}
\newcommand{\bfe}{\mathbf{e}}
\newcommand{\BR}{\operatorname{BR}}
\newcommand{\doubleR}{\mathbb{R}}

\title{On the Convergence of Fictitious Play: A Decomposition Approach\footnote{Copyright \href{https://ijcai-22.org/}{International Joint Conferences on Artificial Intelligence Organization} (IJCAI-ECAI-22 Conference). All rights reserved.}}
\author{
Yurong Chen$^1$\and
Xiaotie Deng$^1$\footnote{Corresponding Author}\and
Chenchen Li$^2$\and
David Mguni$^3$\and\\
Jun Wang$^4$\and
Xiang Yan$^5$\And
Yaodong Yang$^6$
\affiliations
$^1$Center on Frontiers of Computing Studies, School of Computer Science, Peking University \\
$^2$Business Growth BU, JD.com
$^3$Huawei R\&D
$^4$UCL
$^5$Huawei TCS Lab\\
$^6$Institute for AI, Peking University
\emails
chenyurong@pku.edu.cn, 
xiaotie@pku.edu.cn, 
lcc104@qq.com, 
davidmguni@hotmail.com, 
jun.wang@cs.ucl.ac.uk,
xyansjtu@163.com, 
yaodong.yang@pku.edu.cn
}

\begin{document}

\maketitle

\begin{abstract}
Fictitious play (FP) is one of the most fundamental game-theoretical learning frameworks for computing Nash equilibrium in $n$-player games; it builds the foundation for modern multi-agent learning algorithms. 
Although FP has provable convergence guarantees on zero-sum games and potential games, many real-world problems are often a mixture of both and the convergence property of FP has not been fully studied yet. 
In this paper, we extend the convergence results of FP on the combinations of such games and beyond. 
Specifically, we derive new conditions for FP to converge by leveraging game decomposition techniques.
We further develop a linear relationship unifying cooperation and competition in the sense that these two classes of games are mutually transferable. 
Finally, we offer analysis for the non-convergent example of FP, the Shapley game, and give sufficient conditions for  FP to   converge.
\end{abstract}

 \section{Introduction}
 \label{sec: introduction}

Solving Nash equilibrium (NE) \cite{nash1950equilibrium,deng2021complexity} in multi-player games has become a central interest in a variety of fields including but not limited to economics, computer science and artificial intelligence.
Among the many NE solvers, fictitious play (FP)   \cite{browniterative} is one of the most well-known learning algorithms. 
In FP, at each iteration, each player takes a best response to the empirical average of the opponent's previous strategies. 
It is guaranteed that FP dynamics converge to an NE on two-player zero-sum games \cite{Robinson1951iterative} and potential games \cite{Monderer1996ficitious,monderer1996potential,spg-david}, with no need to access the other player's utility information.
Thus, the design principle of FP (i.e., the iterative best-response dynamics) has inspired many other approximation solutions to NE.  
For example, in solving two-player zero-sum games, 
two representative methods are double oracle (DO) \cite{mcmahan2003planning,dinh2021online} and policy space response oracle (PSRO) \cite{lanctot2017unified,feng2021neural,perez2021modelling,liu2021unifying} where a subgame NE is adopted as the best-responding target and multi-agent reinforcement learning (MARL) algorithms \cite{yang2020overview} are applied to approximate the best response.     
Similarly, \citeauthor{heinrich2015fictitious} \shortcite{heinrich2015fictitious}
combined fictitious self-play with deep RL methods and demonstrated remarkable performance on Leduc Poker and Limit Texas Holdem at real-world scale. Besides its inspirations for modern MARL algorithms \cite{muller2020a,yang2018mean}, FP itself has been shown to have good performance of converging to approximate equilibria on some more general games \cite{candogan2013dynamics,Ostrovski2014payoff},
and is still a popular interest of research to the communities of both game theory and machine learning \cite{swenson2018best}.

However, the assumptions of zero-sum or potential games are rather limited. 
As \citeauthor{dasgupta2019survey} \shortcite{dasgupta2019survey} mentioned, adversarial learning has been modeled as a two-player zero-sum game, but chances are that the learner's loss may not equal the adversary's utility. 
Similar as zero-sum games characterizing full competition, potential games can be regarded as full cooperation.
While in most real applications, there are both competition and cooperation among players. 
For example, in a market, sellers for the same type of goods are not only players competing for the same group of buyers, but also collaborators attracting more buyers to the market.
More general results for FP to converge on combinations of competition and cooperation are needed, considering its current limited positive convergence results. 

 \noindent\textbf{Our techniques and results.} We leverage two game decomposition techniques, Hodge Decomposition~\cite{candogan2011flows} and Strategic Decomposition~\cite{hwang2020strategic}, to study the convergence of FP on mixtures of games modelling full competition and cooperation. The idea of game decomposition is to treat the set of games as a linear space and decompose a game into several simple basis games, whose NEs are easy to characterize. 
 Interestingly, in both decompositions, a game is made up of a competitive part, a cooperative part and a trivial part. It is known that FP converges on all these basis games. Combining two decompositions enables us to find latent relationships among games and study how a game's dynamics is influenced by each component.
 We note that, however, there will be no convergence guarantee for FP on arbitrary combinations of these basis games, since they span the whole game space, and FP has been proved to fail to converge on all  games \cite{shapley1964some}. 

The contributions of this paper are as follows, 

\begin{itemize}
    \item We prove that FP converges on any linear combination of a harmonic game (competition) and a potential game (cooperation), so long as they sum to be strategically equivalent to a zero-sum game or to an identical interest game. The conditions are polynomial-time checkable. 
   \item   We show that, utilizing a linear parameter, games lying in these two equivalent classes can be transformed from one class into another.
   \item We give a new analysis of the non-convergence of Continuous-time FP (CFP) on the classic example, the Shapley game, from the view of dynamical system and game decomposition, and provide a sufficient condition about initial conditions for it to converge on linear combinations of zero-sum games and identical interest games. 

\end{itemize}



Since many machine learning methods are built on FP, the fact that FP converges in larger game classes provides a guarantee for them to be applied to more real situations safely.
We argue that the classes of games considered in our paper are of wide interest and include many important games like load balancing games~\cite{vocking2007selfish}, cost and market sharing games with various distribution rules~\cite{gopalakrishnan2011characterizing}, and strictly competitive games~\cite{adler2009note}. 

\noindent\textbf{Related works. } 
Our analysis is mainly based on the techniques of game decomposition. To make the paper self-contained, we discussed related work on game decompositions here, and put the full related work in the Appendix. 

\citeauthor{candogan2011flows} \shortcite{candogan2011flows} and \citeauthor{hwang2020strategic} \shortcite{hwang2020strategic} studied decompositions for finite games. 
Helmholtz Decomposition \cite{balduzzi2018mechanics}, on the other hand, studied the decomposition of continuous games into potential games and Hamiltonian games.
But few works are about making use of decomposition to analyze games.
To the best of our knowledge, utilizing game decomposition techniques to analyze game dynamics is still a new area. \citeauthor{tuyls2018symmetric} \shortcite{tuyls2018symmetric} decomposed an asymmetric bimatrix game, where two players have the same number of pure actions, into two one population games, from the evolutionary game theory point of view. \citeauthor{cheung2020chaos} \shortcite{cheung2020chaos} utilized the canonical decomposition of decomposing a game into a zero-sum game and an identical interest game to study the chaotic behaviors of general-sum games under multiplicative weight dynamics. 
We study the convergence properties of fictitious play. Instead of checking which component dominates the other, we model the mixture of two basis games by linear combinations and see how the game patterns evolve when the parameter changes smoothly and continuously, which can also be seen as a linear homotopy from the view of homotopy method.

\noindent\textbf{Paper organization.} 
Section \ref{sec:backgroud} introduces necessary preliminaries.
Section \ref{sec:hodge} shows our main convergence results. The illustration of transformations between cooperation and competition is shown in Example \ref{equivalent example} and Section \ref{coopetition}. 
Section \ref{sec:shapley} analyzes the Shapley game and give a condition for CFP to converge.
Section \ref{sec:conclusion} summarizes and give future work.


\section{Background}\label{sec:backgroud}

We use bold lowercase characters to denote vectors. $\mathbf{1}_n$ and $\mathbf{0}_n$ denote $n$-dimensional all-one and all-zero vectors, respectively, $\mathbf{e}_i$ the vector the $i^{\text{th}}$ coefficient of which is 1 and all other coefficients are 0. 
$[n]$ denotes the set $\{1,\dots,n\}$. 

\subsection{Games and Nash equilibrium}
We focus on two-player bimatrix games. 
We use bold uppercase letters to denote games and uppercase letters to denote matrices: A game $\mathbf{G}$ is given in the bimatrix form $\left(A, B\right)$, 
where the first (second) matrix in the pair denotes the payoff matrix of player $1$ (player $2$, respectively). 
Both matrices have dimension $m\times n$, i.e., player $1$ (player $2$) has $m$ ($n$, respectively) actions. 

We call an action $i$ a pure strategy, and a distribution over all actions a mixed strategy. Denote the set of all mixed strategies as $\Delta_m$, where 
$\Delta_{m} \coloneqq \{\bfp \in \doubleR^m :$ $ p_{i} 
\geq 0, \forall i \in\{1, \ldots, m\}, \sum_{i=1}^{m} p_{i}=1\}$.
Given a game $\mathbf{G}=(A,B)$ and two mixed strategies $\bfp\in \Delta_m$, $\bfq \in \Delta_n$, player 1 (player 2)'s utility is $\bfp^\top A \bfq$ ($\bfp^\top B \bfq$,  respectively). We use $\BR_i(\cdot,\mathbf{G})$ to denote player $i$'s best response set: 
\begin{equation*}
    \resizebox{.999\linewidth}{!}{$
            \displaystyle
            \BR_1(\bfq, \mathbf{G})=\argmax_{i\in [m]} ( A\bfq)_i,  \BR_2(\bfp,\mathbf{G})=\argmax_{j\in[n]} (\bfp^\top B)_j. 
        $}
\end{equation*}
We omit the last variable of $\BR_i$ when there is no confusion. 

A Nash equilibrium (NE) is a pair of mixed strategies such that no one wants to deviate with the other's strategy fixed:

\begin{definition} Strategy pair $\left(\bfp^{*}, \bfq^{*}\right)$ is a Nash Equilibrium (NE) of game $\mathbf{G}=(A,B)$ if for any $\bfp \in \Delta_{m}$ and $\bfq \in \Delta_{n}$,
\begin{equation*}
\bfp^{* \top} A \bfq^{*} \geq \bfp^{\top} A \bfq^{*},~~~~\bfp^{* \top} B \bfq^{*} \geq \bfp^{* \top} B \bfq 
\end{equation*}
\end{definition}
We call an NE pure when strategies in the NE are all pure strategies, and mixed otherwise.  

Let $\mathcal{G}$ be the set of all bimatrix games. Given games $\mathbf{G}_1=\left(A_1,B_1\right)$,  $\mathbf{G}_2=\left(A_2,B_2\right) \in \mathcal{G}$, define their addition to be $\mathbf{G}_1+\mathbf{G}_2 := \left(A_1+A_2, B_1+B_2\right)$. Given a scalar $\alpha \in \doubleR$ and a game $\mathbf{G}=\left(A,B\right)$, define the scalar multiplication to be $\alpha \mathbf{G}=\left (\alpha A,\alpha B\right )$. 
Now $\mathcal{G}$ is a linear space and we can consider the combinations and decompositions of games. 

\subsection{Basic games and relations among games}\label{sub: gamede}

We introduce subspaces of $\mathcal{G}$ that are basic in this paper. 

\begin{definition}
\label{def of basic games}
	Define the following subspaces of $\mathcal{G}$
	
\noindent 1.  \textbf{Identical interest} games, $ \mathcal{I}:= \left\{ (A,B) \in \mathcal{G}: A=B \right\}$

\noindent 2.  \textbf{Zero-sum} games, $\mathcal{Z}:= \left\{ (A,B) \in \mathcal{G}:  A+B=0\right\}$

\noindent 3. \textbf{Non-strategic} games, $\mathcal{E}:=\left\{ (A,B) \in \mathcal{G}: A = \mathbf{1}_{m} \mathbf{u}^{\top}, \right.~~$ $\left.B  = \mathbf{v}\mathbf{1}^\top_{n}, \mathbf{u}\in \doubleR^n, \mathbf{v} \in \doubleR^m  \right\}
$

\noindent 4. \textbf{Normalized} games, $
		\mathcal{N}:= \{ (A,B) \in \mathcal{G}: \sum^m_{j=1} A_{ji}=0, ~\forall~i\in [n]; \sum^n_{j=1}B_{ij}=0, ~\forall~i\in[m] \}$
\end{definition}

Identical interest games and zero-sum games are important games in game theory and machine learning, especially in multi-agent learning and adversarial learning, with the former modeling team cooperation and the latter modeling competition. 
In a non-strategic game,  
a player's utility only depends on the other player's strategy, and is not affected by her own strategy at all.
Thus any strategy pair of the game is an NE. 

Noticing that adding a non-strategic game to a game does not change its original game structure, e.g., the best response structure and NEs,
one can define an equivalence relation between two games if they only differ by a non-strategic game (also called strategic equivalence by \citeauthor{hwang2020strategic} \shortcite{hwang2020strategic}), and mainly focus on the normalized games. In a normalized game, the sum of one player's utilities, with her own strategy changing and the other's fixed, equals to zero. 

Different from \citeauthor{hwang2020strategic} \shortcite{hwang2020strategic}, we call two games $\mathbf{G}$ and $\mathbf{G}^\prime$ are \textbf{additionally equivalent} if $\mathbf{G}^\prime = \mathbf{G} + \mathbf{E}$ for some $\mathbf{E}\in \mathcal{E}$. Here we introduce a more general equivalence among games.

\begin{definition}
\label{def: se}
    Game $\mathbf{G}=(A,B)$ is \textbf{strategically equivalent} to game $\mathbf{G}^\prime=(A^\prime, B^\prime)$, if there exist two positive constants $\alpha,~\beta \in \doubleR_{+}$ and a non-strategic game $\mathbf{E}\in \mathcal{E}$ such that
\begin{equation*}
(A^\prime,B^\prime)=(\alpha A,\beta B)+\mathbf{E}, ~\text { for some } \mathbf{E} \in \mathcal{E}
\end{equation*}
\end{definition}
While additional equivalence is a special case of strategic equivalence, the notion of additional equivalence is compatible with the operations of the space. Lemma \ref{lem:strategic} shows that both equivalences preserve the best response structure. 
\begin{restatable}{lemma}{lemstrategic}\label{lem:strategic}
    Given two strategically equivalent games $\mathbf{G}$ and $\mathbf{G}^\prime$, we have
    \begin{equation*}
         \BR_1(\bfq,\mathbf{G})=\BR_1(\bfq,\mathbf{G}^\prime),~~~~
        \BR_2(\bfp,\mathbf{G})=\BR_2(\bfp,\mathbf{G}^\prime)
    \end{equation*}
\end{restatable}
We use $\mathcal{S(\cdot)}$ ($\mathcal{A(\cdot)}$) to denote the set of games that are strategically (additionally) equivalent to the games in $\cdot$. 
In particular, we call games in $\mathcal{S(Z)}$ ($\mathcal{S(I)}$) zero-sum equivalent games (identical interest equivalent games, respectively).

\subsection{Discrete-time fictitious play (DFP)}\label{sub: fp}

In fictitious play, each player regards the empirical distribution over the other player's actions as her belief towards the other player's mixed strategy, and acts myopically to maximize her utility in the next step. Specifically, let $\bfp(t)\in \Delta_m$ and $\bfq(t)\in\Delta_n$ be the beliefs of two players' strategies at time step $t$, then the sequence $(\bfp(t),\bfq(t))$ is a discrete-time fictitious play (DFP) if:
\begin{equation*}
    (\bfp(0),\bfq(0))\in \Delta_m \times \Delta_n
\end{equation*}
and for all $t$: 
\begin{equation}
\label{dfp update}
\begin{aligned}
\bfp(t+1)&\in \frac{t}{t+1}\bfp(t)+\frac{1}{t+1}\BR_1(\bfq(t)),\\
    \bfq(t+1)&\in \frac{t}{t+1}\bfq(t)+\frac{1}{t+1}\BR_2(\bfp(t))
\end{aligned}
\end{equation}
With a specific tie-breaking rule, we can regard $\BR_i(\cdot)$ as a vector and the inclusion above becomes equality.  

In the following sections, DFP is simply called FP. With abuse of notation, the term FP refers to the dynamic rules in Eqn.\ (\ref{dfp update}) or the sequences generated by the rules, according to the context. We say a game has fictitious play property (FPP) if every FP sequence of it converges. Noticing that games with the same best response structure enjoy the same FPP.

\begin{figure*}[htbp]
\centering
\begin{align*}
\begin{bmatrix} 
-14, 21 & -20, 30 & -14, 21 \\ 
18, -27 & 14, -21 & 2, -3 \\ 
-18,27  & 0, 0 & -16, 24 
\end{bmatrix}
=&
\left(2\begin{bmatrix} 
-7 & -10 & -7 \\ 
9 & 7 & 1 \\ 
-9 & 0 & -8 
\end{bmatrix},
3\begin{bmatrix} 
7 & 10 & 7 \\ 
-9 & -7 & -1 \\ 
9 & 0 & 8  
\end{bmatrix}\right)
+
\begin{bmatrix} 
0, 0 & 0, 0 & 0, 0 \\ 
0, 0 & 0, 0 & 0, 0 \\ 
0, 0 & 0, 0 & 0, 0
\end{bmatrix}
\end{align*}
\caption{Game $\mathbf{G}\in \mathcal{S(Z)}$: $\mathbf{G}=(2Z,-3Z)$, where $Z\in \doubleR^{m\times n}$. 
}
\label{form of example}
\end{figure*}

\section{Convergence on the combinations of cooperation and competition}\label{sec:hodge}

In this section, we formally present our first two main results. In Section \ref{sec:proof}, we introduce two game decompositions, and present our results on the linear combinations of basis games. In Section \ref{coopetition}, we illustrate how the players' relationships transform with respect to the linear parameter. 

\subsection{Proof of convergence by game decompositions}\label{sec:proof}
Before formally stating our first result,
we present two important game decompositions: 
\begin{theorem}[Strategic Decomposition]~\cite{hwang2020strategic}
\label{thm: strategic}
 The space of games $\mathcal{G}$ can be decomposed as: 
\begin{equation*}
    \mathcal{G} = (\mathcal{I} \cap \mathcal{N})   \oplus (\mathcal{Z} \cap \mathcal{N}) \oplus \mathcal{B}.
\end{equation*}
where $\mathcal{B} \coloneqq (\mathcal{I} + \mathcal{E}) \cap (\mathcal{Z} + \mathcal{E})$ is the set of zero-sum equivalent potential games, $\oplus$ denotes the direct sum of two linear subspaces.
\end{theorem}
$\mathcal{I+E}$ is the space of games additionally equivalent to identical interest games, and is actually the space of all potential games. The equivalence between this definition and the one using potential function is shown in Appendix. 
As shown by \citeauthor{hwang2020strategic} \shortcite{hwang2020strategic}, a two-player zero-sum equivalent potential game $\mathbf{B}\in \mathcal{B}$ has the form
\begin{equation}
\label{form of b}
    \displaystyle
\mathbf{B}= (\mathbf{u}\mathbf{1}^\top_n, \mathbf{1}_m\mathbf{v}^\top)+\mathbf{E}
= \left(\mathbf{u}\mathbf{1}^\top_n+\mathbf{1}_m\mathbf{x}^{\top}, \mathbf{1}_m\mathbf{v}^\top+\mathbf{y}\mathbf{1}^{\top}_{n}\right)
\end{equation} 
for some $\mathbf{u},\mathbf{y}\in \doubleR^m$, $\mathbf{v},\mathbf{x}\in \mathbb{R}^n$ and $\mathbf{E}= (\mathbf{1}_m\mathbf{x}^{\top},\mathbf{y}\mathbf{1}^{\top}_{n})\in \mathcal{E}$. That is, a player's utility is not affected by the other's strategy, which can be seen as the opposite of non-strategic games. Each player has a dominant strategy, and a pure NE exists.

\begin{theorem}[Hodge Decomposition] \cite{candogan2011flows}
\label{thm: hodge}
The space of games $\mathcal{G}$ can be decomposed as:
\begin{equation*}
\mathcal{G} = \mathcal{P} \oplus \mathcal{H} \oplus \mathcal{E}.    
\end{equation*}
where $\mathcal{P}\coloneqq\mathcal{N}\cap \mathcal{\left(I+E\right)}$ denotes normalized potential games, and $\mathcal{H}\coloneqq \left \{(A,B)\in \mathcal{N}: mA+nB=0 \right\}$ normalized harmonic games. 
$\mathcal{P+E}$ is the set of all potential games. $\mathcal{H+E}$ is the set of all harmonic games.
\end{theorem}

The definition of normalized harmonic games tells that they are like zero-sum games. Thus both decompositions show that any bimatrix game is made up of a fully cooperative component, a fully competitive component and a component that either has both features or is trivial. Combining the decompositions and equivalences makes us able to study bimatrix games from multiple angles. 

Since the basis games in the decompositions generate the whole game space, FP will not converge on all combinations of them (a famous example of non-convergence is the Shapley game \cite{shapley1964some}). Notice that on all basis games, however,
FP will converge:
results
on zero-sum games and potential games are known \cite{Monderer1996ficitious,Robinson1951iterative}.
When the tie-breaking rule is decided, best responses in FPs are always the same on games in $\mathcal{B}$ and $\mathcal{E}$ --- a dominant strategy in $\mathbf{B}\in\mathcal{B}$ and the prescribed strategy by the tie-breaking rule in $\mathbf{E}\in\mathcal{E}$ --- thus FP will converge. 
In Appendix, we provide a simple proof to show that FP also converges on harmonic games.

It is then interesting to study under what conditions do combinations of these games preserve FPP. 
By first conducting experiments on mixtures of normalized harmonic games and normalized potential games, 
we find out that if they are components of a zero-sum game, then FP converges on any linear combination of them. 
The following theorem gives the formal explanation for this phenomenon and provide a more general condition for FP to converge: if their sum is either fully competitive or cooperative, then any linear combination of these games has FPP. 
Recall that the set $\mathcal{S(\cdot)}$ is the set of games strategically equivalent to games in $~\cdot~$, then we have:
\begin{restatable}{theorem}{thmcombine}
\label{thm: combine}
For any game $\mathbf{G} \in \mathcal{S(Z)\cup S(I)}$ with Hodge Decomposition, 
	\begin{equation*}
	    \mathbf{G}=\mathbf{P}+\mathbf{H}+\mathbf{E}
	\end{equation*}
	where $\mathbf{P}\in \mathcal{P}$, $\mathbf{H}\in \mathcal{H}$, $\mathbf{E}\in\mathcal{E}$.  
Then for any $\lambda \in \mathbb{R}$, game $\lambda \mathbf{P}+(1-\lambda)\mathbf{H}$ has FPP. 
\end{restatable}

We note that the decomposition of a game that is either strategically equivalent to a zero-sum game or to an identical interest game is non-trivial: It can have all game components, since a game lying in these two classes does not necessarily belong to any basis game class. One can refer to the decomposition of the game in Example \ref{equivalent example} in Appendix C.

\begin{figure*}[htbp]
    \begin{subfigure}{0.5\textwidth}
    \centering
    \includegraphics[width=0.67\linewidth]{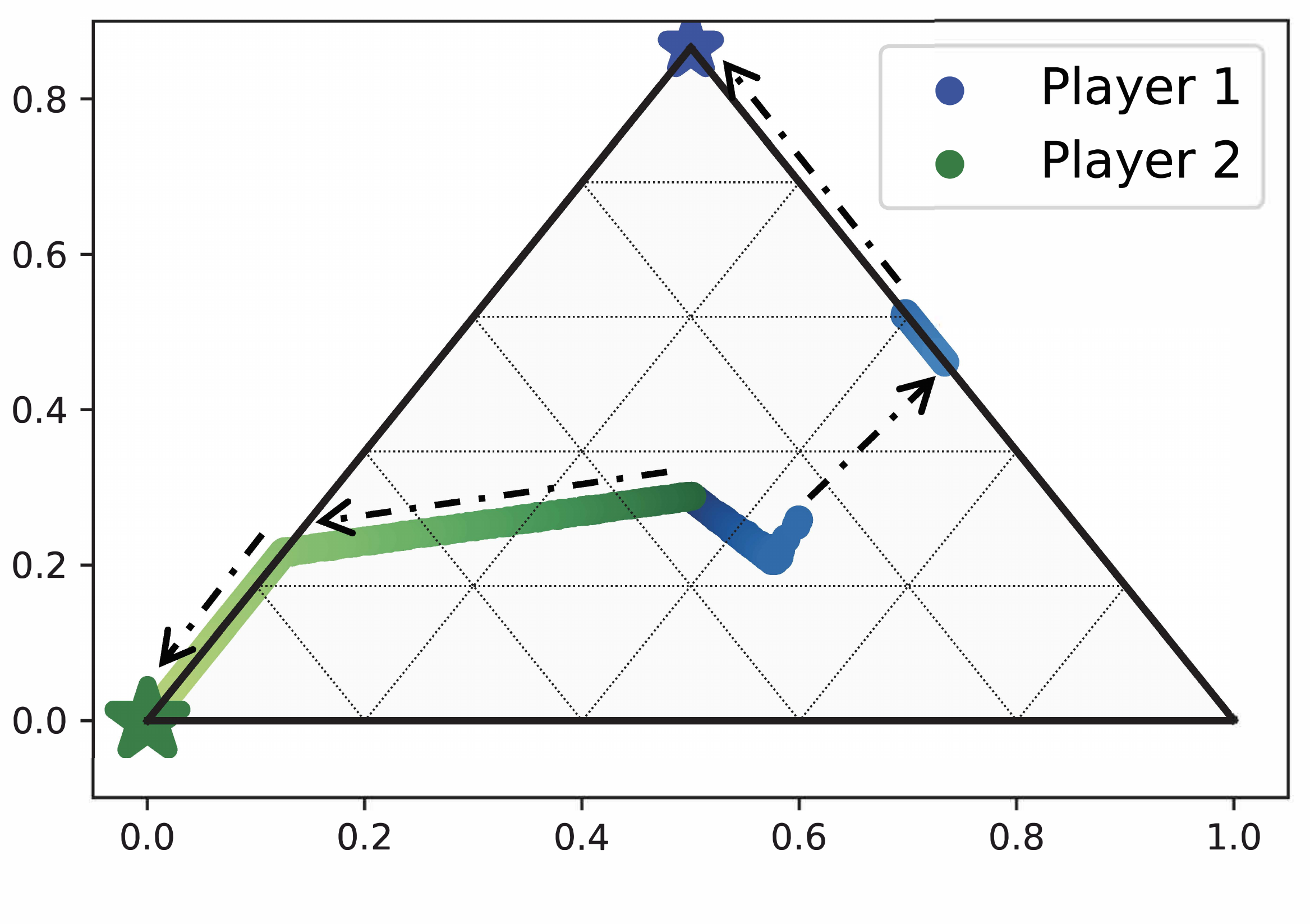}
    \caption{Equilibrium changes in one strategy simplex. }
    \label{fig:barycentric}
    \end{subfigure}
    \begin{subfigure}{.5\textwidth}
    \centering
    \includegraphics[width=0.67\linewidth]{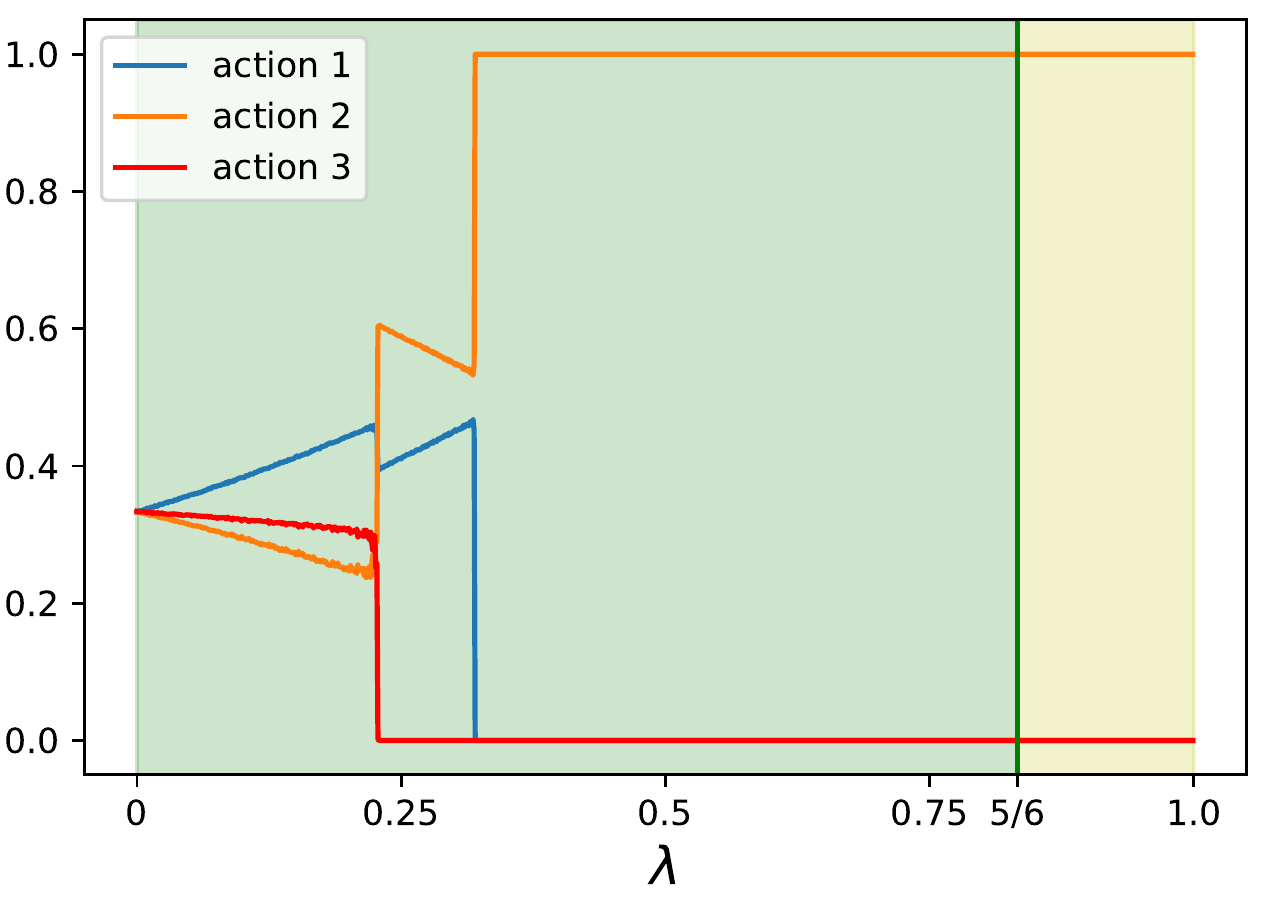}
    \caption{Equilibrium strategy changes of player 1. }
    \label{fig:linechart}
    \end{subfigure}
\caption{How Nash equilibrium of game $\lambda \mathbf{P}+(1-\lambda)\mathbf{H}$ changes as the linear parameter $\lambda$ increases from $0$ to $1$. 
(a) The changes of equilibrium strategies of both players are shown in one strategy simplex. The black arrows tell the direction of the changes. The stars denote the stopping point. When $\lambda$ is small, the harmonic part dominates. The game shows more competitive patterns, as the equilibrium strategies are mixed in the interior of the strategy simplex. As $\lambda$ gets larger, the support of equilibrium strategies shrinks. 
(b) The changes of player 1's strategy in a line chart. In the green area (when $\lambda<\frac{5}{6}$), the game is zero-sum equivalent. In the yellow area (when $\lambda>\frac{5}{6}$), the game is identical interest equivalent. When $\lambda$ exceeds $\frac{5}{6}$, the game is more cooperative, and the equilibrium stays pure.
}
\label{fig:example}
\end{figure*}

To prove Theorem \ref{thm: combine}, we first give a necessary lemma. Lemma \ref{se close zero-sum potential} states the properties of zero-sum equivalent potential games;
it shows that sets $\mathcal{S(I)}$ and $\mathcal{S(Z)}$ are closed under the operation of adding a zero-sum equivalent potential game. 



\begin{restatable}{lemma}{secloseb}
\label{se close zero-sum potential}
If $\mathbf{G} \in \mathcal{S(Z)}\cup \mathcal{S(I)}$, then $\mathbf{G}+\mathbf{B} \in \mathcal{S(Z)}\cup \mathcal{S(I)}$ for any $\mathbf{B}\in\mathcal{B}$. 
\end{restatable}

With Lemma \ref{se close zero-sum potential}, we can now prove Theorem \ref{thm: combine}. 
The convergence of games $\lambda \mathbf{P} + (1-\lambda)\mathbf{H}$ are actually the trade-offs between cooperation ($\mathbf{P}$) and competition ($\mathbf{H}$). 
Though the decompositions of games considered are non-trivial, by further decomposing the components using the other decomposition, we find interesting relations among components of different decompositions, which is a key technique in our proof. 

\begin{proof}[Proof sketch of Theorem \ref{thm: combine}]
 We show the proof sketch of the case when $\mathbf{G}\in \mathcal{S(Z)}$ here and the full proof is in Appendix. 
Let $\mathbf{P}=\mathbf{I}+\mathbf{E}=(I,I)+\mathbf{E}$, where $\mathbf{I} \in \mathcal{I}$, $\mathbf{E}\in \mathcal{E}$, $I\in \doubleR^{m\times n}$. $\mathbf{H}$ can be formulated as $(nZ,-mZ)$ for some $Z \in \doubleR^{m\times n}$.

When $\mathbf{G}=(Z^\prime, -\alpha Z^\prime) + \mathbf{E}^\prime \in \mathcal{S(Z)}$, where $Z^\prime \in \doubleR^{m\times n}$, $\alpha >0$, $\mathbf{E}^\prime \in \mathcal{E}$, then we have
\begin{equation*}
    (I,I)+(nZ,-mZ) = (Z^\prime, -\alpha Z^\prime) + \mathbf{E}^\prime - \mathbf{E}
\end{equation*}
By letting $\mathbf{E}^\prime - \mathbf{E}= (E_1, E_2)$, where $E_1, E_2 \in \doubleR^{m \times n}$ satisfies $E_1=\mathbf{1}_m\mathbf{u}^\top$,  $E_2=\mathbf{v}\mathbf{1}^\top_n$, for some $\mathbf{u} \in \doubleR^n$, $\mathbf{v} \in \doubleR^m$, 
we have $I$ and $Z$ represented as linear combinations of $Z^\prime$, $E_1$ and $E_2$. Thus for any $\lambda \in \doubleR$,
\begin{equation*}
    \begin{aligned}
    &\lambda \mathbf{P} + (1-\lambda) \mathbf{H} 
   = \left((a_1(\lambda)Z^\prime,b_1(\lambda)Z^\prime\right)\\ 
   +&\left(a_2(\lambda)E_1+a_3(\lambda)E_2, b_2(\lambda)E_1+b_3(\lambda)E_2\right) + \lambda\mathbf{E}
    \end{aligned}
\end{equation*}
By the definition of non-strategic games and Eqn.\ (\ref{form of b}), $\left(a_2(\lambda)E_1+a_3(\lambda)E_2, b_2(\lambda)E_1+b_3(\lambda)E_2\right) \in \mathcal{B}$ is a zero-sum equivalent potential game. When $a_1(\lambda)b_1(\lambda)\ne 0$, game  $(a_1(\lambda) Z^\prime, b_1(\lambda)Z^\prime) \in \mathcal{S(Z)}\cup \mathcal{S(I)}$. 
By Lemma \ref{se close zero-sum potential}, $\lambda \mathbf{P}+ (1-\lambda)\mathbf{H} \in \mathcal{S(I)}\cup \mathcal{S(Z)}$ and thus has FPP.  

When $a_1(\lambda)b_1(\lambda)=0$, $\lambda \mathbf{P} + (1-\lambda)\mathbf{H}$ has one payoff matrix in the form of $\mathbf{x}\mathbf{1}^\top_n+\mathbf{1}_m\mathbf{y}^{\top}$ for some $\mathbf{x} \in \doubleR^m$, $\mathbf{y}\in \doubleR^n$. The player with this kind of payoff matrix has a dominant strategy. During each time step of FP, the player will choose her dominant strategy, and the other best responds to that dominant strategy. The sequence will converge to a pure NE. 

\end{proof}
Following Theorem \ref{thm: combine}, define $\mathcal{D}$ to be a new set of games, in which one of the payoff matrices has the form of $\mathbf{x}\mathbf{1}^\top_n+\mathbf{1}_m\mathbf{y}^{\top}$ for some $\mathbf{x} \in \doubleR^m$, $\mathbf{y}\in \doubleR^n$, then: 

\begin{corollary}
\label{coro:new game d}
For any game $\mathbf{G}\in \mathcal{D}$ with Hodge Decomposition 
	\begin{equation*}
	    \mathbf{G}=\mathbf{P}+\mathbf{H}+\mathbf{E}
	\end{equation*}
	where $\mathbf{P}\in \mathcal{P}$, $\mathbf{H}\in \mathcal{H}$, $\mathbf{E}\in\mathcal{E}$. 
	Then for any $\lambda \in \mathbb{R}$, game $\lambda \mathbf{P}+(1-\lambda)\mathbf{H}$ has FPP. 
\end{corollary}

Conditions in Theorem \ref{thm: combine} and Corollary \ref{coro:new game d} can be checked in linear time, e.g. through the method by \citeauthor{MR4051409} \shortcite{MR4051409}.

\subsection{Transformations from cooperation to competition}
\label{coopetition}

We use one example to show the non-trivial decompositions of games considered in Section \ref{sec:proof} and how the proportions of the cooperative component and the competitive component influence the game, when $\lambda \in [0,1]$ ranges from $0$ to $1$. 

\begin{example}\label{equivalent example}
A game in $\mathcal{S(Z)}$ is shown in Figure \ref{form of example}. 
We decompose this game using Hodge Decomposition into three components $\mathbf{P}\in\mathcal{P}$, $\mathbf{H}\in\mathcal{H}$ and $\mathbf{E}\in \mathcal{E}$, the details of which are in Appendix C,  and compute the game $\mathbf{G}(\lambda)=\lambda \mathbf{P} + (1-\lambda)\mathbf{H}$. The changes of NE which FP converges to with $\lambda$ changing from $0$ to $1$ is shown in Figure \ref{fig:example}. In Figure \ref{fig:barycentric}, the strategy trajectories are presented in the mixed strategy simplex: each vertex of the triangle is an action $\mathbf{v}_i$. 
We draw mixed strategy $\bfp =(p_1,p_2,p_3)\in \Delta_3$ on the convex combination of the vertices, $\sum^3_{i=1}p_i\mathbf{v}_i$. 
The green dots represent the strategy trajectory of player 1, and the blue dots that of player 2. The star sign denotes the strategies when $\lambda=1$. We set the step length to be $0.001$. For each $\lambda$, we run FP starting from $\left(\mathbf{v}_1,\mathbf{v}_1\right)$ for 500,000 rounds. The strategy trajectories of both players start at the center of a strategy simplex, i.e.\ the uniform equilibrium of harmonic games \cite{candogan2011flows}. As $\lambda$ increases, the strategies move towards the boundaries, and their supports become small. When $\lambda$ is large enough, they reach a pure NE (nodes with stars), which is typical of potential games, and no longer moves. Figure \ref{fig:linechart} shows the changes of player 1's strategies in line chart.

We utilize the algorithm by \citeauthor{MR4051409} \shortcite{MR4051409} to decide when will $\mathbf{G(\lambda)}\in \mathcal{S(Z)}$.
When $\lambda < \frac{5}{6}$, The competitive part takes over, $\mathbf{G}(\lambda) \in \mathcal{S(Z)}$. When $\lambda >\frac{5}{6}$, $\mathcal{G(\lambda)}\in\mathcal{S(I)}$. When $\lambda = \frac{5}{6}$, $\mathbf{G}(\lambda)\not\in \mathcal{S(Z)\cup S(I)}$ but belongs to $\mathcal{D}$ instead. We draw the point $\lambda = \frac{5}{6}$ on Figure \ref{fig:linechart}. One can find out that when $\lambda$ is less than but close to $\frac{5}{6}$, the equilibrium for FP to converge to already becomes pure and never moves even when $\lambda$ exceeds $\frac{5}{6}$. 


\end{example}

When $\lambda$s in the above examples reach the threshold, the games enter $\mathcal{D}$ instead of $\mathcal{S(Z)}\cap\mathcal{S(I)}$. We argue that games in $\mathcal{D}$ show the same dynamic patterns with games in $\mathcal{S(Z)\cap S(I)}$. Specifically, we show that $\mathcal{S(Z)\cap S(I)}=\mathcal{A(Z)\cap A(I)}=\mathcal{B}$, and players in games belonging to $\mathcal{B}$ have dominant strategies. In both classes of games, the players will choose a determined strategy and the dynamics converge to the corresponding pure NE directly. 

\begin{restatable}{proposition}{sameintersection}
\label{same intersection}
$
\mathcal{S(Z)}\cap\mathcal{S(I)}=\mathcal{A(Z)}\cap \mathcal{A(I)} = \mathcal{B}. 
$
\end{restatable}

The process of transforming from zero-sum equivalent games to identical interest equivalent games with the changes of a linear parameter can be seen as an instinct feature of these two games, and a bridge linking non-cooperative game-theoretic view of cooperation and competition. Using a linear parameter that changes with time, we may be able to model how the relationships of players evolve in real world. Furthermore, when $\lambda$ is small enough, we can regard the potential part in $\lambda \mathbf{P} + (1-\lambda) \mathbf{H}$ as a small perturbation added up to a harmonic game. Our result shows that when the perturbation satisfies certain conditions, the game with perturbations can still be regarded as a zero-sum equivalent game.


\section{Analysis on the non-convergence example}\label{sec:shapley}

In this section, we use game decompositions to give a new analysis of the classic non-convergent example of FP, the Shapley game~\cite{shapley1964some}.
We consider Continuous-time Fictitious Play (CFP), a useful tool to give insights into DFP's dynamics.  we present the Shapley game and introduce a novel function, the best response utility function. We give two results about initial conditions based on this function. 

\begin{figure*}[!htbp]

\centering

\begin{align*}
\underbrace{\begin{bmatrix} 
0, 0 & 2, 1 & 1, 2 \\ 
1, 2 & 0, 0 & 2, 1 \\ 
2,1  & 1, 2 & 0, 0 
\end{bmatrix}}_{\text{the Shapley game}}
=
\underbrace{\begin{bmatrix} 
-1, -1 & 0.5, 0.5 & 0.5, 0.5 \\ 
0.5, 0.5 & -1, -1 & 0.5, 0.5 \\ 
0.5, 0.5 & 0.5, 0.5 & -1, -1 
\end{bmatrix}}_{\substack{\text{Potential game}\\\text{Normalized identical interest game}}}
+
\underbrace{\begin{bmatrix} 
0, 0 & 0.5, -0.5 & -0.5, 0.5 \\ 
-0.5, 0.5 & 0, 0 & 0.5, -0.5 \\ 
0.5, -0.5 & -0.5, 0.5 & 0, 0
\end{bmatrix}}_{\substack{\text{Harmonic game}\\ \text{Normalized zero-sum game}}}
+
\underbrace{\begin{bmatrix} 
1, 1 & 1, 1 & 1, 1 \\ 
1, 1 & 1, 1 & 1, 1 \\ 
1, 1 & 1, 1 & 1, 1 
\end{bmatrix}}_{\substack{\text{Non-strategic game}\\\text{Zero-sum equivalent potential game}}}
\end{align*}
\caption{
Hodge Decomposition and Strategic Decomposition of the Shapley game. It is a coincident that they are the same.
}\label{eqn:shapley}
\end{figure*}

\subsection{Continuous-time fictitious play (CFP)}
\label{cfp}

DFP can be regarded as an update procedure for players' empirical beliefs, where the update rate is $1$.
Now consider the corresponding continuous version, by rescaling the rate to $\delta>0$ and letting $\delta\rightarrow 0$.
This equivalently defines the derivatives of the sequence $(\bfp(t),\bfq(t))$ with respect to $t$ : \footnote{One can further rescale the variable $t$ to $t'=\ln{t}$ (for $t>0$) to eliminate the the denominator.}
\begin{equation}\label{def:cfp}
\resizebox{.91\linewidth}{!}{$
    \displaystyle
	\dot{\bfp}(t) = \frac{\BR_1(\bfq(t)) - \bfp(t) }{t},  \dot{\bfq}(t) = \frac{\BR_2(\bfp(t)) - \bfq(t)}{t}.$
}
\end{equation}
The detailed derivation from DFP to CFP is in Appendix D. 

For CFP, we define another stable state following the studies on dynamical systems.


\begin{definition}[Cycle]\label{def:cycle}
	A CFP follows a cycle $\mathrm{C}$ if there is an integer $K>0$ and a sequence of $K$ pairs of pure strategies:
	\begin{equation*}
	\mathrm{C} = \left\{ \left(i_1, j_1\right), \dots, \left(i_K, j_K\right) \right\} 
	\end{equation*}
s.t., $\exists T>0$, $\forall t>T$, $(\BR_1(\bfq(t)),\BR_2(\bfp(t)))$ of the play takes the values $\left(i_1, j_1\right), \dots, \left(i_K, j_K\right)$ periodically. 
\end{definition}


It is a special case that a CFP converges to an NE when it follows a cycle.
So far, all the known convergence results of CFP and DFP are the same. 
Intuitively, when discrete time step tends to $+\infty$, so as the denominator of Eqn.\ (\ref{dfp update}), the changes of $(\bfp(t),\bfq(t))$ at each time step become infinitesimal, and will resemble the derivatives defined by Eqn.\ (\ref{def:cfp}). 
Thus, CFP can provide useful insights into original DFP, though the relationships of two dynamics are still not clear.

Now we study the conditions that a CFP with some initial conditions converges to an NE.

\subsection{The Shapley game}
\label{shapley game}
The Shapley game, proposed by \citeauthor{shapley1964some} \shortcite{shapley1964some}, 
is a counter-example of FP's convergence: 
except under 
certain
initial conditions, 
FP does not converge to NE. 
Figure \ref{eqn:shapley} shows its payoff matrices and the two decomposition results.

When two players' initial strategies are different, e.g.\ 
$(1, 2)$, FP follows cycle
$\mathrm{C}_1 = \{(1, 2), (1, 3), (2, 3),(2, 1),  $ $(3, 1), (3, 2) \}$, but does not converge: 
both players' strategies will change periodically but never reach any fixed point.
When their initial strategies are the same, e.g.\ $(1, 1)$, FP follows another cycle
$\mathrm{C}_2 = \{(2, 2), (1, 1), (3, 3)\}$,
and tends to an NE $(\bfp,\bfq)$, where $\bfp = \bfq = (\frac{1}{3}, \frac{1}{3}, \frac{1}{3})$.

\subsection{The best response utility function}
\label{bruf}
We further consider more general combinations of the decomposed components $\mathbf{G}(\lambda)=\lambda \textbf{P} + (1-\lambda)\textbf{H}=\lambda \textbf{I} + (1-\lambda)\textbf{Z}$. 
Define the best response utility (BRU) function for game $\mathbf{G}=(A,B)$ and players' strategies $(\bfp,\bfq)$ to be
\begin{equation*}
    U(\bfp,\bfq, \mathbf{G}) = \max_{i\in[m]} (A\bfq)_i +\max_{j\in [n]}(\bfp B)_j.
\end{equation*}
It is the sum of maximal utilities each player can get by her best response to the opponent's strategy. We omit the last variable of $U$ when there is no confusion. Experimental analysis (in Appendix D) on the Shapley game shows that whether DFP converges to NE or not, $U$ always converges.

For the empirical frequencies $(\bfp(t),\bfq(t))$ ($t\geq 1$) obtained by CFP, consider the derivative of $U(\bfp(t),\bfq(t))$ ($U(t)$ for short) with respect to $t$, for almost all $t$,
then for any fixed $t_0>0$, we have
\begin{equation}
\label{eqn:solve_u}
U(t) = \frac{t_0U(t_0)}{t} +\frac{1}{t} \int^t_{t_0}  G_{i(\tau),j(\tau)}\mathrm{d}\tau,
\end{equation}
where $G=A+B$, $i(\tau)$ and $j(\tau)$ are short for the index of $\BR_1(\bfq(\tau))$ and $\BR_2(\bfp(\tau))$. The detailed derivation is in Appendix. 
This leads to a sufficient condition for the convergence of BRU.
\begin{restatable}{lemma}{lembruconverge}

\label{lem:bru_converge}
For a game $\mathbf{G}=(A,B)$, let $G=A+B$. Assume CFP follows a cycle $\{\left(i_1, j_1\right), \dots, \left(i_K, j_K\right)\}$.
If $G_{i_1, j_1} = \dots = G_{i_K, j_K}$, then BRU converges.
\end{restatable}


Back to the Shapley game, on cycle $\mathrm{C_2}$, 
$G_{i_1, j_1} =\dots = G_{i_3,j_3} = \min{G_{ij}}$. 
To be compared, on cycle $\mathrm{C_1}$,  $G_{i_1, j_1}=\dots = G_{i_6,j_6} > \min{G_{ij}}$. 
This means BRU converges but not to the minimal possible value, which coincides with the fact that the strategies fail to converge.

One can further find out that when a CFP starts with the same strategies for both players, i.e., when it converges to NE, by the symmetry of the game, $\mathbf{p}(t)=\mathbf{q}(t)$ for all $t$. Then for all strategy pairs $(\mathbf{p}(t),\mathbf{q}(t))$ related to $\mathrm{C}_2$, $\dot{U}$ satisfies
\begin{equation*}
\dot{U} = 2 \mathbf{e}(t) I \mathbf{e}(t) - (\bfp^\top B \mathbf{e}(t) + \mathbf{e}(t) A \bfq) < 0, \end{equation*}
where $I$ is the payoff matrix of the Shapley game's normalized identical interest component, and $\mathbf{e}(t) = \BR_1(\bfq(t))=\BR_2(\bfp(t))$. $U(t)$ will keep decreasing till it tends to NE. 

Now we consider the Shapley game as a linear combination of a normalized identical interest game and a normalized zero-sum game and we state the reason formally in Theorem \ref{thm:lyapnuov} why it converges under such initial conditions. Given a cycle $\mathrm{C}$ of game $\mathbf{G}$ such that there exist CFPs tending to $\mathrm{C}$, let $P(\mathrm{C},\mathbf{G})$ be all the mixed strategy pairs that lie on the paths following $\mathrm{C}$, and $\operatorname{Conv}(\mathrm{S})$ be the closed convex hull of set $\mathrm{S}$. Denote the set of $\mathbf{G}$'s NEs as $\mathcal{X}(\mathbf{G})$. 
We have

\begin{restatable}{theorem}{thmlyapunov}
\label{thm:lyapnuov}
For game $\mathbf{G}(\lambda)=\lambda \mathbf{I} + \mathbf{Z}$ for some $\mathbf{I}\in (\mathcal{I} \cap \mathcal{N})$ and $\mathbf{G}\in (\mathcal{Z} \cap \mathcal{N})$, $\lambda \in \mathbb{R}$. If from an initial point, CFP on $\mathbf{G}(\lambda)$ enters the same cylce $\mathrm{C}$ of $\mathbf{G}(0)$, 
where $\mathrm{C}$ is a cycle that a convergent CFP on $\mathbf{G}(0)$ will tend to, and for all $(\bfp,\bfq)\in \operatorname{Conv}(P(\mathrm{C},\mathbf{G}))\backslash\mathcal{X}(\mathbf{G}(0))$, $\dot{U}(\bfp, \bfq) < 0$, 
then CFPs on $\mathbf{G}(\lambda)$ which tend to $\mathrm{C}$ will converge to an NE. 

\end{restatable}


Theorem \ref{thm:lyapnuov} gives a sufficient condition for a CFP to converge. Intuitively, when $|\lambda|$ is small enough, the dynamics of  $\mathbf{G}(\lambda)$ will resemble that of $\mathbf{G}(0)$. If CFPs on these games enter the same cycle of $\mathbf{G}(0)$, and $\dot{U}<0$ for all points related to this cycle, then CFPs on $\mathbf{G}(\lambda)$ converge to NE. 


\section{Conclusion and future work}\label{sec:conclusion}

Decomposing the game space into combinations of simple classes enables us to find the new relations among games.
In this paper, we use this method to prove a new condition for FP to converge and build a bridge between games modelling full competition and cooperation. 
We derive an instinct property for them that these two classes of games can be mutually transformed to the other with a simple parameter. This ability of mutual transformation may be applied in the dynamic multi-agent systems in which agents' relationships vary with time, which 
helps study complex real environments. Furthermore, we analyze the well-known example of FP's non-convergence, the Shapley game, and give a sufficient condition for its continuous version to converge. 

As for the future work, the first one is to have more analysis on FP dynamics using game decomposition techniques. We note that it will be of vital importance but great challenge to give a full characterization, and it is impossible to have a convergence guarantee on arbitrary combinations of games. Another interesting problem is to study the dynamical properties of simple games, e.g., zero-sum games, with perturbations. 


 \appendix

 \input{suppl.tex}




\small{
\section*{Acknowledgments}
This work is supported by Science and Technology Innovation 2030 –“The Next Generation of Artificial Intelligence” Major Project No. (2018AAA0100901).
\bibliographystyle{named}
\bibliography{ijcai22}
}


\end{document}

%% file: suppl.tex
\section{Full related work in Section \ref{sec: introduction}}

Fictitious Play is the first learning algorithm designed to calculate the value of a zero-sum game \cite{brown1949some,browniterative}.
It has been proved to converge on two-player zero-sum games \cite{Robinson1951iterative}, identical interest games \cite{Monderer1996ficitious}, potential games\cite{monderer1996potential},  $2\times 2$ games \cite{miyasawa1961convergence,metrick1994fictitious}, and  $2\times n$ games \cite{Berger2005fictitious} and games solvable by iterative elimination \cite{nachbar1990evolutionary}.
In comparison with their works, ours shows FP's convergence on general $m\times n$ games, under moderate conditions. 

The dynamic behaviors of FP on general-sum games are very complicated.
\citeauthor{krishna1998convergence} \shortcite{krishna1998convergence} proved that, for almost all cases (except a set with zero Lebesgue measure), FP converges in the form of a robust cycle, where the support size of the equilibrium is 2.
On the contrast, chaotic behaviors in FP were found when there are three or more strategies for each player \cite{cowan1993dynamical,richards1997geometry,van2011fictitious}.
Motivated by these studies, we study CFP in the analysis of the Shapley game using the idea in dynamical systems. 

Our analysis is mainly based on the techniques of game decomposition. 
\citeauthor{candogan2011flows} \shortcite{candogan2011flows} and \citeauthor{hwang2020strategic} \shortcite{hwang2020strategic} studied Hodge Decomposition and Strategic Decomposition for finite games, respectively.
Helmholtz Decomposition \cite{balduzzi2018mechanics}, on the other hand, studied the decomposition of continuous games into potential games and Hamiltonian games.
But few works are about making use of decomposition to analyze games.
To the best of our knowledge, utilizing game decomposition techniques to analyze game dynamics is still a new area. \citeauthor{tuyls2018symmetric} \shortcite{tuyls2018symmetric} decomposed an asymmetric bimatrix game, where two players have the same number of pure actions, into two one population games, from the evolutionary game theory point of view. \citeauthor{cheung2020chaos} \shortcite{cheung2020chaos} utilized the canonical decomposition of decomposing a game into a zero-sum game and an identical interest game to study the chaotic behaviors of general-sum games under multiplicative weight dynamics. 
We study the convergence properties of fictitious play. Instead of checking which component dominates the other, we model the mixture of two basis games by linear combinations and see how the game patterns evolve when the parameter changes smoothly and continuously, which can also be seen as a linear homotopy from the view of homotopy method. 

Another type of popular algorithms for solving games is no-regret learning. Similarly to FP, no-regret learning is also used as the base algorithm to find NE in zero-sum games, such as Poker AI \cite{brown2019superhuman}. 
In both FP and no-regret learning methods, each agent updates her policy separately to optimize individual payoff or regret, thus they can approach NE of zero-sum games, but have no guarantee for general-sum games. Besides, FP is not no-regret \cite{cesa2006prediction}. Although no-regret learning can converge to coarse correlated equilibria (CCE), CCE is not necessarily equal to NE in general-sum games, and is not our focused solution concept.

\section{Proofs of Lemma \ref{lem:strategic} in Section \ref{sec:backgroud}}
\lemstrategic*
\begin{proof}
    It suffices to prove the first equation, and the second equation is similar. 
    
    A non-strategic game can be expressed as $(\mathbf{1}_{m} \mathbf{u}^{\top}, \mathbf{v} \mathbf{1}_{n}^{\top})$, for some $\mathbf{u}\in \mathbb{R}^n$, $\mathbf{v}\in \mathbb{R}^m$. Then $A^\prime = A + \mathbf{1}_{m} \mathbf{u}^{\top}$. For any $\mathbf{p}\in\Delta_m$, $\mathbf{q}\in \Delta_n$, 
    \begin{equation*}
        \mathbf{p}^\top (\mathbf{1}_m\mathbf{u}^\top)\mathbf{q}=(\mathbf{p}^\top\mathbf{1}_m)(\mathbf{u}^\top\mathbf{q})=\mathbf{u}^\top\mathbf{q}
    \end{equation*}
    Given $\mathbf{q}\in \Delta_n$, and any $\mathbf{p} \in \operatorname{BR}_1(\mathbf{q},\mathbf{G})$, we have
    \begin{align}
        &\mathbf{p}^\top A \mathbf{q} &\geq& (\mathbf{p}^\prime)^\top A \mathbf{q}, ~\forall~\mathbf{p}^\prime \in \Delta_m \nonumber\\
        \Longleftrightarrow ~ &\mathbf{p}^\top \alpha A \mathbf{q}+\mathbf{u}^\top\mathbf{q} &\geq& (\mathbf{p}^\prime)^\top \alpha A \mathbf{q}+\mathbf{u}^\top\mathbf{q} \nonumber\\
       \Longleftrightarrow ~ & \mathbf{p}^\top A^\prime \mathbf{q} &\geq& (\mathbf{p}^\prime)^\top A^\prime \mathbf{q} \nonumber
    \end{align}
    Thus, $\mathbf{p} \in \operatorname{BR}_1(\mathbf{q},\mathbf{G})\Leftrightarrow \mathbf{p} \in \operatorname{BR}_1(\mathbf{q},\mathbf{G}^\prime)$, which completes the proof of the first equation. 
\end{proof}

\section{Proofs in Section \ref{sec:hodge}}


\subsection{The full version of Strategic decomposition}

\begin{theorem}
\label{app_thm: strategic}
(Strategic Decomposition~\cite{hwang2020strategic}) The space of games $\mathcal{G}$ can be decomposed as: 
\begin{enumerate}
	\item $\mathcal{G} = (\mathcal{I} \cap \mathcal{N})   \oplus (\mathcal{Z} + \mathcal{E})$;
	\item $\mathcal{G} = (\mathcal{I} + \mathcal{E})   \oplus (\mathcal{Z} \cap \mathcal{N})$;
	\item $\mathcal{G} = (\mathcal{I} \cap \mathcal{N})   \oplus (\mathcal{Z} \cap \mathcal{N}) \oplus \mathcal{B}$.
\end{enumerate}
where $\mathcal{B} \coloneqq (\mathcal{I} + \mathcal{E}) \cap (\mathcal{Z} + \mathcal{E})$ is the set of zero-sum equivalent potential games, $\oplus$ denotes the direct sum of two linear subspaces.

\end{theorem}

\subsection{The equivalence between two definitions of potential games}

Another definition of potential game is: A game $\mathbf{G}=\left(A,B\right)$ is a potential game if there is a potential function $\phi$, represented by matrix $I\in \doubleR^{m\times n}$ s.t. $\forall j,j^\prime \in [n], i,i^\prime\in[m]$, $I_{ij}-I_{i^\prime j}=A_{ij}-A_{i^\prime j}$, $I_{ij}-I_{ij^\prime}=B_{ij}-B_{i j^\prime}$.
which means the potential function catches the incentives for players to change their strategies. 

For any fixed $i_0\in [m]$ and $j_0\in[n]$, $\forall j \in [n], i \in [m]$, $A_{ij}=I_{ij}+A_{i_0 j}-I_{i_0 j}$, $B_{ij}=I_{ij}+B_{i j_0}-I_{i j_0}$, this means $\mathbf{G}=(I,I)+((A_{i_0 j}-I_{i_0 j})_{ij},(B_{i j_0}-I_{i j_0})_{ij})$, i.e. a sum of an identical interest and a non-strategic game. 

For any game $\mathbf{P}=\mathbf{I}+\mathbf{E}\in \mathcal{I+E}$, where $\mathbf{I}=(I,I)\in\mathcal{I}$, $\mathbf{E}\in\mathcal{E}$, it can be easily verified that $I$ is the potential function mentioned in the above definition. 

\subsection{Proof that any harmonic game has FPP}
\begin{proposition}
Every harmonic game has FPP.
\end{proposition}

\begin{proof}
Any normalized harmonic game $\mathbf{H}=(A,B)$ satisfies 
\begin{equation*}
    mA+nB=0
\end{equation*}
Let $\mathbf{H}^\prime =(A^\prime, B^\prime)$, s.t.\ $A^\prime = mA$, $B^\prime = nB$, then $\mathbf{H}^\prime$ is a zero-sum game and strategically equivalent to $\mathbf{H}$, thus FP converges on $\mathbf{H}$. Since adding a  non-strategic component does not change the best response dynamics of the game, FP converges on any harmonic game. 
\end{proof}

\subsection{A proposition for better understanding}

\begin{proposition}
\label{harmonic decomposition}
For any 
normalized 
harmonic game $\mathbf{H}\in \mathcal{H}$, $\mathbf{H}=\mathbf{I}+\mathbf{Z}$, where $\mathbf{I}\in (\mathcal{I}\cap \mathcal{N})$, $\mathbf{Z}\in (\mathcal{Z}\cap \mathcal{N})$,  
and 
$\lambda \mathbf{I+Z}$ has FPP, for any $\lambda \in \mathbb{R}$. 
\end{proposition}

Proposition \ref{harmonic decomposition} shows that the convergence of FP on harmonic games is actually the result of trade-off between cooperation and competition. 
 The idea is that in harmonic games, when the fraction of identical interest part, $\lambda$, increases, the cooperative part will occupy and $\lambda \mathbf{I} + \mathbf{Z}$ is an identical interest equivalent game. When $\lambda$ decreases, the competitive part will take over. Similar idea is used in the proof of our main theorem. 


To prove Proposition \ref{harmonic decomposition}, we first prove a Lemma. From Strategic Decomposition, a normalized potential game only has normalized identical interest component and zero-sum equivalent potential component. Lemma \ref{lem:harmonic} tells that a normalized harmonic game actually only has normalized identical interest component and normalized zero-sum component. That is, $\mathcal{P}\subset (\mathcal{I\cap N})+(\mathcal{Z\cap N})$.

\begin{lemma}
\label{lem:harmonic}
A normalized harmonic game $\mathbf{H}$ can be decomposed uniquely into a normalized identical interest game and a normalized zero-sum game. 
\end{lemma}

\begin{proof}
The main idea of the proof is to show that a basis of $\mathcal{H}$ belongs to $\left(\mathcal{I\cap N}\right)\oplus \left(\mathcal{Z\cap N}\right)$.
This means in the decomposition of Theorem \ref{thm: strategic}, the basis of $\mathcal{H}$ has no zero-sum equivalent potential component. 
Then as any normalized harmonic game is a linear combination of the basis, by the closure of linear spaces, all normalized harmonic games have no zero-sum equivalent potential component. 

For any $i\in \left\{1,\dots,m-1\right\}$, $j\in \left\{1,\dots,n-1\right\}$, define bimatrix games $\mathbf{G}^{ij}=(nA^{ij},-mA^{ij})$, where $A^{ij}\in \mathbb{R}^{m\times n}$ such that 
$$
A_{k l}^{i j}=\left\{\begin{aligned}
1 & ~~~~~\text { if }(k, l)=(i, j) \text { or }(k, l)=(i+1, j+1); \\
-1 & ~~~~~\text { if }(k, l)=(i+1, j)  \text { or }(k, l)=(i, j+1); \\
0 & ~~~~~\text { otherwise. }
\end{aligned}\right.
$$
Then $\left\{\mathbf{G}^{ij}\right\}$ is a basis of $\mathcal{H}$ \cite{candogan2011flows}. 
To find the three components of its Strategic Decomposition, consider the projection methods proposed by \citeauthor{hwang2020strategic} \shortcite{hwang2020strategic}.
Precisely, denote its identical interest component as $\mathbf{I}^{ij}$, zero-sum equivalent potential component as $\mathbf{B}^{ij}$ and normalized zero-sum component as $\mathbf{Z}^{ij}$ respectively, we have
\begin{equation*}
    \begin{aligned}
      \mathbf{I}^{ij}&=\frac{n-m}{2}(A^{ij},A^{ij})\\
      \mathbf{B}^{ij}&=0\\
      \mathbf{Z}^{ij}&=\frac{m+n}{2}\left(A^{ij},-A^{ij}\right)
    \end{aligned}
\end{equation*}
\end{proof}

\begin{proof}[Proof of Proposition \ref{harmonic decomposition}]
Let $\mathbf{H}=\mathbf{I+Z}=\left(I,I\right)+\left(Z,-Z\right)$, by the definition of normalized harmonic games, we have
\begin{equation*}
m(I+Z)+n(I-Z)=0
\end{equation*}
This means
\begin{equation*}
I=\frac{n-m}{m+n}Z
\end{equation*}

Then for any $\lambda\in \mathbb{R}$, 

\begin{equation*}
\small{
 \begin{aligned}
 \mathbf{G}_\lambda \coloneqq& \lambda \mathbf{I}+\mathbf{Z} = \left(\lambda I,\lambda I\right)+\left(Z,-Z\right)\\
 = &\left(\frac{\left(\lambda+1\right)n-\left(\lambda-1\right)m}{m+n}Z,\frac{\left(\lambda-1\right)n-\left(\lambda+1\right)m}{m+n}Z\right)
 \end{aligned}}
 \end{equation*}

If $n=m$, $\mathbf{G}_\lambda$ is always $\mathbf{Z}$ and we obtain the convergence result directly.
Thus w.l.o.g., assume $n>m$. 
When $\lambda$ satisfies $\left(\lambda-1\right)n-\left(\lambda+1\right)m>0$, $\mathbf{G}_\lambda$ is strategically equivalent to an identity interest game $\left(Z,Z\right)$. When $\left(\lambda-1\right)n-\left(\lambda+1\right)m<0$, $\mathbf{G}_\lambda$ is strategically equivalent to a  zero-sum game $\left(Z,-Z\right)$. When $\left(\lambda-1\right)n-\left(\lambda+1\right)m=0$, $\mathbf{G}_\lambda$ becomes $\left(\frac{\left(\lambda+1\right)n-\left(\lambda-1\right)m}{m+n}Z, 0\right)$ and the convergence is trivial. 
\end{proof}

\subsection{Proof of Lemma \ref{se close zero-sum potential}}

\secloseb*

\begin{proof}
We only prove that when $\mathbf{G}\in \mathcal{S(Z)}$, $\mathbf{G}+\mathbf{B} \in \mathcal{S(Z)}$. The case when $\mathbf{G} \in \mathcal{S(I)} $ is similar. 

By the definition of non-strategic games and Eqn. (\ref{form of b}), 
\begin{equation}
\label{eqn:zero-sum equivalent potential}
\begin{aligned}
 \mathbf{B}&=(\mathbf{u}\mathbf{1}^\top_n, \mathbf{1}_m\mathbf{v}^\top)+(\mathbf{1}_m\mathbf{x}^{\top},\mathbf{y}\mathbf{1}^{\top}_{n})\\ 
 &= \left(\mathbf{u}\mathbf{1}^\top_n+\mathbf{1}_m\mathbf{x}^{\top}, \mathbf{1}_m\mathbf{v}^\top+\mathbf{y}\mathbf{1}^{\top}_{n}\right),  
\end{aligned}
\end{equation}
for some $\mathbf{u}, \mathbf{y}\in \mathbb{R}^m$, $\mathbf{v},\mathbf{x}\in \mathbb{R}^n$. 

Since $\mathbf{G}=(\alpha Z, -\beta Z)$ for some $\alpha, \beta >0$ and $Z \in \doubleR^{m\times n}$, then
\begin{equation*}
\begin{aligned}
\mathbf{G} + \mathbf{B} =&\left(\alpha \left(Z+\left(\frac{1}{\alpha}\mathbf{u}\right)\mathbf{1}^\top_n+\mathbf{1}_m\left(\frac{1}{\alpha}\mathbf{x}\right)^{\top}\right),\right.\\
&\left.\phantom{(}\beta\left(-Z+\mathbf{1}_m\left(\frac{1}{\beta}\mathbf{v}\right)^\top+\left(\frac{1}{\beta}\mathbf{y}\right)\mathbf{1}^{\top}_{n} \right)\right)
\end{aligned}
\end{equation*}
is strategically equivalent to $\mathbf{Z}+\mathbf{B}^\prime$,  for zero-sum game $\mathbf{Z}=(Z,-Z)$ and zero-sum equivalent potential game $\mathbf{B}^\prime = \left(\left(\frac{1}{\alpha}\mathbf{u}\right)\mathbf{1}^\top_n+\mathbf{1}_m\left(\frac{1}{\alpha}\mathbf{x}\right)^{\top},\mathbf{1}_m\left(\frac{1}{\beta}\mathbf{v}\right)^\top+\left(\frac{1}{\beta}\mathbf{y}\right)\mathbf{1}^{\top}_{n}\right)$, and thus additionally equivalent to $\mathbf{Z}+\mathbf{Z}^\prime$ for $\mathbf{Z}, \mathbf{Z}^\prime \in \mathcal{Z}$. This means $\mathbf{G}+\mathbf{B}$ is strategically equivalent to a zero-sum game. 
\end{proof}

\subsection{Full Proof of Theorem \ref{thm: combine}}

\thmcombine*

\begin{proof}
Let $\mathbf{P}=\mathbf{I}+\mathbf{E}=(I,I)+\mathbf{E}$, where $\mathbf{I} \in \mathcal{I}$, $\mathbf{E}\in \mathcal{E}$, $I\in \doubleR^{m\times n}$. $\mathbf{H}$ can be formulated as $(nZ,-mZ)$ for some $Z \in \doubleR^{m\times n}$.

When $\mathbf{G}=(Z^\prime, -\alpha Z^\prime) + \mathbf{E}^\prime \in \mathcal{S(Z)}$, where $\alpha >0$, $\mathbf{E}^\prime \in \mathcal{E}$, then we have
\begin{equation*}
    (I,I)+(nZ,-mZ) = (Z^\prime, -\alpha Z^\prime) + \mathbf{E}^\prime - \mathbf{E}
\end{equation*}

By letting $\mathbf{E}^\prime - \mathbf{E}= (E_1, E_2)$ where $E_1, E_2 \in \doubleR^{m \times n}$ satisfying $E_1=\mathbf{1}_m\mathbf{u}^\top$,  $E_2=\mathbf{v}\mathbf{1}^\top_n$, for some $\mathbf{u} \in \doubleR^n$, $\mathbf{v} \in \doubleR^m$, we have the following equation system:
\begin{equation}
    \label{eqn:solve1-supp}
    \begin{aligned}
    I + n Z & = Z^\prime + E_1\\
    I - m Z & = -\alpha Z^\prime + E_2
    \end{aligned}
\end{equation}
By solving linear equation system (\ref{eqn:solve1-supp}), we have
\begin{equation*}
    \begin{aligned}
    I & = \frac{m-n \alpha}{m+n}Z^\prime + \frac{m}{m+n}E_1+\frac{n}{m+n}E_2\\
    Z & = \frac{1+\alpha}{m+n}Z^\prime+\frac{1}{m+n}E_1-\frac{1}{m+n}E_2
    \end{aligned}
\end{equation*}

Thus for any $\lambda \in \doubleR$,  $\lambda \mathbf{P} + (1-\lambda) \mathbf{H}$ equals
\begin{equation*}
    \begin{aligned}
     &\phantom{=}\left(a_1(\lambda)Z^\prime + a_2(\lambda)E_1+a_3(\lambda)E_2, \right. \\
    &\phantom{=(}\left. b_1(\lambda)Z^\prime +b_2(\lambda)E_1+b_3(\lambda)E_2 \right)+\lambda\mathbf{E}\\
   &= \left((a_1(\lambda)Z^\prime,b_1(\lambda)Z^\prime\right)+\left(a_2(\lambda)E_1+a_3(\lambda)E_2,\right.\\
   &\phantom{=((}\left.b_2(\lambda)E_1+b_3(\lambda)E_2\right)+ \lambda\mathbf{E}
    \end{aligned}
\end{equation*}
where 
\begin{equation*}
    \begin{aligned}
    a_1(\lambda) &= \frac{\left(m-n-2n\alpha\right)\lambda+n\left(1+\alpha\right)}{m+n}\\
    a_2(\lambda) &= \frac{\left(m-n\right)\lambda+n}{m+n}\\
    a_3(\lambda) &= \frac{n\left(2\lambda-1\right)}{m+n}\\
    b_1(\lambda) &= \frac{\left(2m+(m-n)\alpha\right)\lambda -m \left(1+\alpha\right)}{m+n}\\
    b_2(\lambda) &= \frac{m\left(2\lambda-1\right)}{m+n}\\
    b_3(\lambda) &= \frac{\left(n-m\right)\lambda+m}{m+n}
    \end{aligned}
\end{equation*}

\begin{figure*}[htbp]
\centering
\begin{align*}
&\begin{bmatrix} 
-14, 21 & -20, 30 & -14, 21 \\ 
18, -27 & 14, -21 & 2, -3 \\ 
-18,27  & 0, 0 & -16, 24 
\end{bmatrix}
=
\left(2\begin{bmatrix} 
-7 & -10 & -7 \\ 
9 & 7 & 1 \\ 
-9 & 0 & -8 
\end{bmatrix},
3\begin{bmatrix} 
7 & 10 & 7 \\ 
-9 & -7 & -1 \\ 
9 & 0 & 8  
\end{bmatrix}\right)
+
\begin{bmatrix} 
0, 0 & 0, 0 & 0, 0 \\ 
0, 0 & 0, 0 & 0, 0 \\ 
0, 0 & 0, 0 & 0, 0
\end{bmatrix}\\
&=
\underbrace{\begin{bmatrix} 
-11, -\frac{4}{3} & -\frac{53}{6}, -\frac{19}{6} & -\frac{73}{6}, \frac{9}{2} \\ 
~\\[-2ex]
\frac{91}{6},-\frac{5}{2} &\frac{101}{6}, -\frac{29}{6} & 18, \frac{22}{3} \\ 
~\\[-2ex]
-\frac{25}{6}, \frac{5}{6} & -8, -7 & -\frac{35}{6}, \frac{37}{6} 
\end{bmatrix}}_{\text{Potential Game}}
+
\underbrace{\begin{bmatrix} 
\frac{5}{3}, -\frac{5}{3} & -\frac{55}{6},\frac{55}{6} & \frac{15}{2}, -\frac{15}{2}\\ 
~\\[-2ex]
\frac{15}{2}, -\frac{15}{2} & -\frac{5}{6}, \frac{5}{6} & -\frac{20}{3}, \frac{20}{3} \\ 
~\\[-2ex]
-\frac{55}{6}, \frac{55}{6} & 10, -10 & -\frac{5}{6}, \frac{5}{6}
\end{bmatrix}}_{\text{Harmonic Game}}
+
\underbrace{\begin{bmatrix} 
-\frac{14}{3}, 24 & -2, 24 & -\frac{28}{3}, 24 \\ 
~\\[-2ex]
-\frac{14}{3}, -17 & -2, -17 & -\frac{28}{3}, -17 \\ 
~\\[-2ex]
-\frac{14}{3}, 17 & -2, 17 & -\frac{28}{3}, 17
\end{bmatrix}}_{\text{Non-strategic Game}} 
\end{align*}
\caption{The hodge decomposition of Example \ref{equivalent example}}
\label{hodge_decomposition_of_example1}
\end{figure*}

By the definition of non-strategic games and Eqn. (\ref{eqn:zero-sum equivalent potential}), $\left(a_2(\lambda)E_1+a_3(\lambda)E_2, b_2(\lambda)E_1+b_3(\lambda)E_2\right) \in \mathcal{B}$ is a zero-sum equivalent potential game. When $a_1(\lambda)b_1(\lambda)<0$, game  $(a_1(\lambda) Z^\prime, b_1(\lambda)Z^\prime) \in \mathcal{S(Z)}$. When $a_1(\lambda)b_1(\lambda)>0$, game  $(a_1(\lambda) Z^\prime, b_1(\lambda)Z^\prime) \in \mathcal{S(I)}$. By Lemma \ref{se close zero-sum potential}, $\lambda \mathbf{P}+ (1-\lambda)\mathbf{H} \in \mathcal{S(I)}\cup \mathcal{S(Z)}$ and thus has FPP.  

When $a_1(\lambda)b_1(\lambda)=0$, $\lambda \mathbf{P} + (1-\lambda)\mathbf{H}$ has one payoff matrix in the form of $\mathbf{x}\mathbf{1}^\top_n+\mathbf{1}_m\mathbf{y}^{\top}$ for some $\mathbf{x} \in \doubleR^m$, $\mathbf{y}\in \doubleR^n$. Then the player with this kind of payoff matrix has a dominant strategy. During each time step of fictitious play, the player will choose her dominant strategy, while the other best response to that dominant strategy, and the sequence will converge to a pure NE.

When $\mathbf{G}=(I^\prime, \alpha I^\prime) + \mathbf{E}^\prime$, where $\alpha >0$, $\mathbf{E}^\prime \in \mathcal{E}$, then we have
\begin{equation*}
    (I,I)+(nZ,-mZ) = (I^\prime, \alpha I^\prime) + \mathbf{E}^\prime - \mathbf{E}
\end{equation*}

By letting $\mathbf{E}^\prime - \mathbf{E}= (E_1, E_2)$ for some $E_1, E_2 \in \doubleR^{m \times n}$, we have the following equation system:
\begin{equation}
    \label{eqn:solve2}
    \begin{aligned}
    I + n Z & = I^\prime + E_1\\
    I - m Z & = \alpha I^\prime + E_2
    \end{aligned}
\end{equation}
By solving linear equation system (\ref{eqn:solve2}), we have

\begin{equation*}
    \begin{aligned}
    I & = \frac{m+n \alpha}{m+n}I^\prime + \frac{m}{m+n}E_1+\frac{n}{m+n}E_2\\
    Z & = \frac{1-\alpha}{m+n}I^\prime+\frac{1}{m+n}E_1-\frac{1}{m+n}E_2
    \end{aligned}
\end{equation*}

Thus for any $\lambda \in \doubleR$, $\lambda \mathbf{P} + (1-\lambda) \mathbf{H} $ equals
\begin{equation*}
    \begin{aligned}
  &\phantom{=}\left(a_1(\lambda)I^\prime + a_2(\lambda)E_1+a_3(\lambda)E_2,\right.\\ 
  &\phantom{=(}\left. b_1(\lambda)I^\prime +b_2(\lambda)E_1+b_3(\lambda)E_2 \right)+\lambda\mathbf{E}\\
   &= \left((a_1(\lambda)I^\prime,b_1(\lambda)I^\prime\right) +\left(a_2(\lambda)E_1+a_3(\lambda)E_2,\right.\\
   &\phantom{=((}\left. b_2(\lambda)E_1+b_3(\lambda)E_2\right) + \lambda\mathbf{E}
    \end{aligned}
\end{equation*}
where 
\begin{equation*}
    \begin{aligned}
    a_1(\lambda) &= \frac{\left(m-n+2n\alpha\right)\lambda+n\left(1-\alpha\right)}{m+n}\\
    a_2(\lambda) &= \frac{\left(m-n\right)\lambda+n}{m+n}\\
    a_3(\lambda) &= \frac{n\left(2\lambda-1\right)}{m+n}\\
    b_1(\lambda) &= \frac{\left(2m-(m-n)\alpha\right)\lambda -m \left(1-\alpha\right)}{m+n}\\
    b_2(\lambda) &= \frac{m\left(2\lambda-1\right)}{m+n}\\
    b_3(\lambda) &= \frac{\left(n-m\right)\lambda+m}{m+n}
    \end{aligned}
\end{equation*}

By the definition of non-strategic games and Eqn. (\ref{eqn:zero-sum equivalent potential}), $\left(a_2(\lambda)E_1+a_3(\lambda)E_2, b_2(\lambda)E_1+b_3(\lambda)E_2\right) \in \mathcal{B}$ is a zero-sum equivalent potential game. By the proof of Theorem \ref{thm: combine}, $\lambda \mathbf{P}+ (1-\lambda) \mathbf{H}$ is either strategically equivalent to an identical interest game or to a zero-sum game. Or it has a pure NE for FP to converge to. 
\end{proof}

\subsection{Proof of Proposition \ref{same intersection}}

\sameintersection*

\begin{proof}
It is obvious that $\mathcal{A(Z)\cap A(I)}\subset \mathcal{S(Z)\cap S(I)}$. 
For any $\mathbf{G}=(A,B) \in \mathcal{S(Z)\cap S(I)}$, $(A,B)=(Z,-\alpha Z)+\mathbf{E}_1 = (I,\beta I) +\mathbf{E}_2$, for some $\alpha,~ \beta >0$, $Z,~ I \in \doubleR^{m\times n}$, $\mathbf{E}_1,~\mathbf{E}_2 \in \mathcal{E}$.  Let $\mathbf{E}_2-\mathbf{E}_1=(E_1,E_2)$, $E_1,~E_2 \in \doubleR^{m\times n}$, we have
\begin{equation}
\begin{aligned}
    Z&=I+E_1\\
    -\alpha Z &= \beta I + E_2
\end{aligned}
\end{equation}
Solving the linear equation systems above, 
\begin{equation}
    \begin{aligned}
    Z&=\frac{\beta}{\alpha+\beta}E_1-\frac{1}{\alpha+\beta}E_2\\
    I&=-\frac{\alpha}{\alpha+\beta}E_1-\frac{1}{\alpha+\beta}E_2
    \end{aligned}
\end{equation}
Then $\mathbf{G}=(Z,-\alpha Z) = (\frac{\beta}{\alpha+\beta}E_1-\frac{1}{\alpha+\beta}E_2, -\frac{\alpha\beta}{\alpha+\beta}E_1+\frac{\alpha}{\alpha+\beta}E_2) +\mathbf{E}_1 \in \mathcal{B}$. 
\end{proof}

\subsection{Hodge Decomposition of Example \ref{equivalent example}}

See in Figure \ref{hodge_decomposition_of_example1}.

\section{Proofs in Section \ref{sec:shapley}}

\subsection{Detailed derivation from DFP to CFP}
Take $\bfp(\cdot)$ as an example:  

In DFP, at each discrete time step, the empirical mixed strategies of the agent are updated according to $\bfp(t+1)= \frac{t}{t+1}\bfp(t)+\frac{1}{t+1}\BR_1(\bfq(t))$. The updated strategy is a weighted linear combination of $\bfp(t)$ and $\BR_1(\bfq(t))$. 

Now consider the continuous version of FP: For {\bf{continuous}} time variable $t\in[1,+\infty)$, let $\delta$ denote the time between two adjustments (like $1$ between $t$ and $t+1$), change the weights $\frac{t}{t+1}$ and $\frac{1}{t+1}$ with $\frac{t}{t+\delta}$ and $\frac{\delta}{t+\delta}$ respectively, we have 
\begin{equation*}
    \bfp(t+\delta)= \frac{t}{t+\delta}\bfp(t)+\frac{\delta}{t+\delta}\BR_1(\bfq(t)).
\end{equation*} 

Subtract $\bfp(t)$ and divide by $\delta$ on both sides, we got 
\begin{equation*}
    \frac{\bfp(t+\delta)-\bfp(t)}{\delta}= \frac{-\bfp(t)+\BR_1(\bfq(t))}{t+\delta}.
\end{equation*}

By letting $\delta\rightarrow 0$, denote the derivatives of $\bfp$ as $\dot{\bfp}$, we got the dynamical systems for CFP: 
\begin{equation*}
    \dot{\bfp}(t) = \frac{\operatorname{BR}_1(\bfq(t)) - \bfp(t) }{t}, ~ \dot{\bfq}(t) = \frac{\operatorname{BR}_2(\bfp(t)) - \bfq(t)}{t}.
\end{equation*}

\subsection{Derivation of Equation \ref{eqn:solve_u}}

In a solution $(\bfp(t),\bfq(t))$ of CFP, by Lemma 4 in  \cite{hofbauer2006best}, for almost all $t>1$, $\dot{U}(t)$ exists and 
\begin{equation}\label{eqn:bru_derivative}
\small{
\begin{aligned}
\dot{U}(t)=&\BR_1(\bfq(t))^\top A \dot{\bfq}(t)+\dot{\bfp}(t)^\top B\BR_2(\bfp(t))\\
=&\BR_1(\bfq(t))^\top A \left(\frac{\BR_2(\bfp(t))-\bfq(t)}{t}\right) \\
& + \left(\frac{\BR_1(\bfq(t))-\bfp(t)}{t}\right)^\top B \BR_2(\bfp(t))\\
=&\frac{1}{t}\left[\BR_1(\bfq(t))^\top (A+B) \BR_2(\bfp(t))\right.\\
&\phantom{\frac{1}{t}[}- \left.\underbrace{\left(\BR_1(\bfq(t))^\top A \bfq(t) + \bfp(t)^\top B\BR_2(\bfp(t)) \right)}_{U(t)}\right]. 
\end{aligned}}
\end{equation}
If $\mathbf{G}\in\mathcal{Z}$, then $A+B=0$ and $\dot{U}(t)$ is always non-positive. If $\mathbf{G}\in \mathcal{B}$, then $\BR_i(\cdot)$ is independent on its variables. 
Thus, $\dot{U}(t)$ is always zero.
Eqn. (\ref{eqn:bru_derivative}) also means
\begin{equation*}
    t\dot{U}+U=\frac{d(tU)}{dt} = \BR_1(\bfq(t))^\top (A+B) \BR_2(\bfp(t)).
\end{equation*}

\begin{figure*}[htbp]
    \begin{subfigure}{0.49\textwidth}
    \centering
    \includegraphics[width=\linewidth]{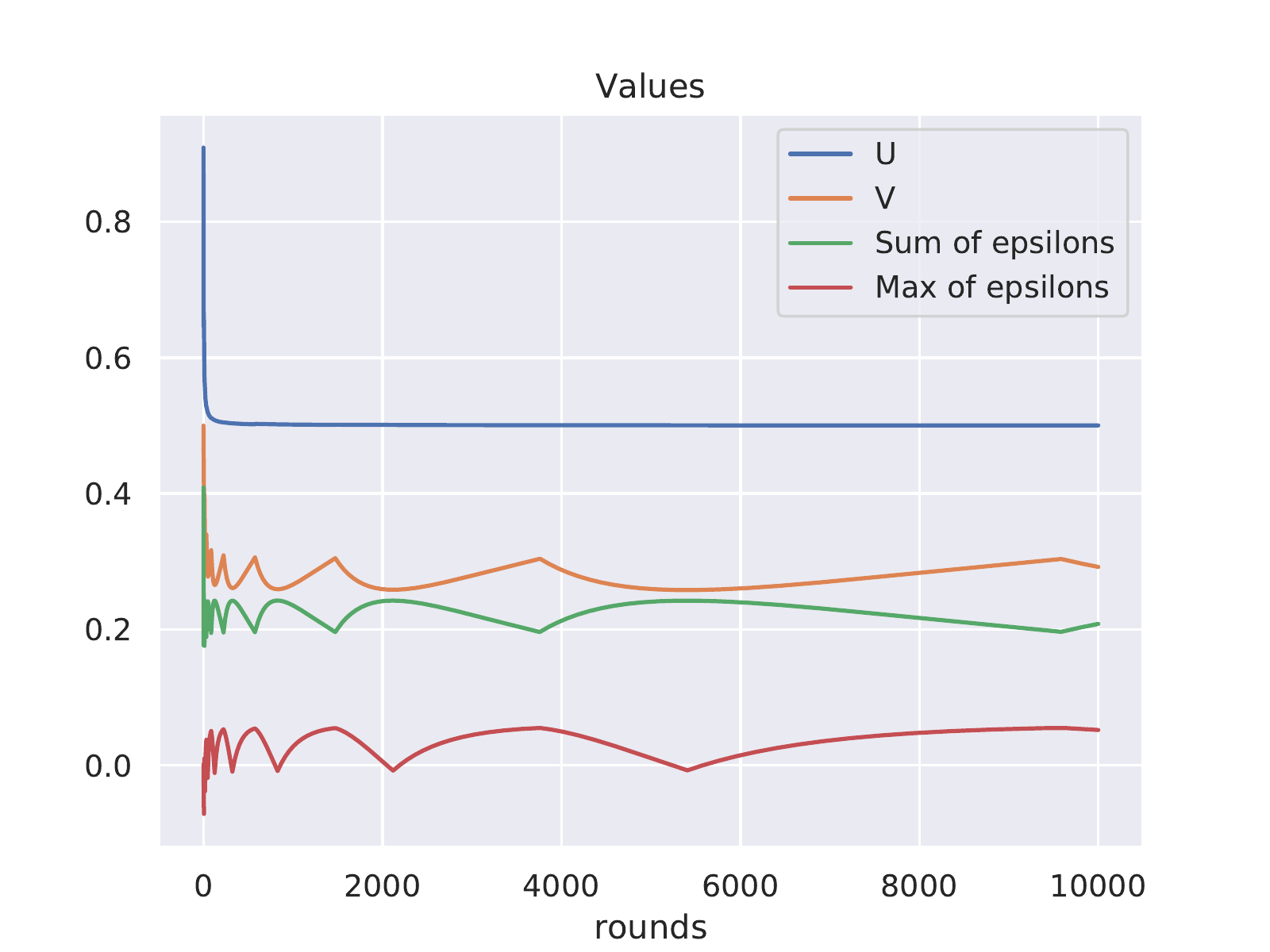}
    \caption{Value changes with iterations on the Shapley game. }
    \label{fig:shapley1}
    \end{subfigure}
    \begin{subfigure}{.49\textwidth}
    \centering
    \includegraphics[width=\linewidth]{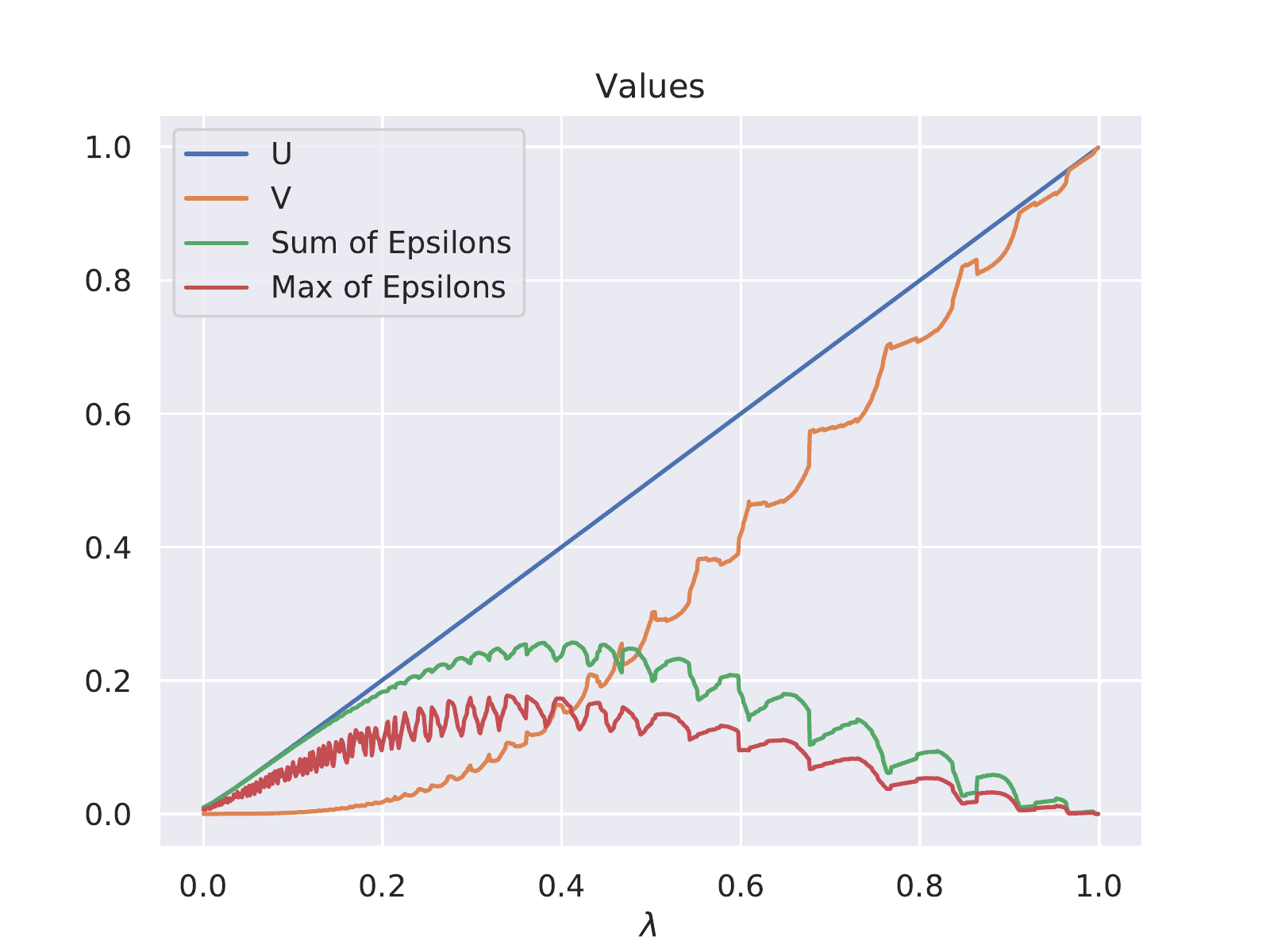}
    \caption{Value changes at termination with $\lambda$ changes.  }
    \label{fig:shapley2}
    \end{subfigure} 
\caption{The results of two experiments on the Shapley game. In two experiments, $U$ shows good smoothness. The overall trend of $ME$ is a one-peak function. The fluctuations on the overall trend can be explained that when $\lambda$ and $\lambda$ is close enoug, the sequences enter the same cycle, but are at different points of the cycle when the algorithm terminates, due to the small difference of $\lambda$. }
\label{fig:shapley analysis}
\end{figure*}

Then for any fixed $t_0>0$, we have
\begin{equation*}
\begin{aligned}
U(t) &= \frac{t_0U(t_0)}{t} +\frac{1}{t} \int^t_{t_0} \BR_1(\bfq(\tau))^\top (A+B) \BR_2(\bfp(\tau)) \mathrm{d}\tau\\
 &\xlongequal{G\coloneqq A+B}\frac{t_0U(t_0)}{t} +\frac{1}{t} \int^t_{t_0}  G_{i(\tau),j(\tau)}\mathrm{d}\tau,
\end{aligned}    
\end{equation*}

\subsection{Proof of Lemma \ref{lem:bru_converge}}

\lembruconverge*

\begin{proof}
Eqn. (\ref{eqn:solve_u}) builds a connection between BRU and the best response pair.
Precisely, if a CFP follows a cycle $\mathrm{C}$ as define in Def. \ref{def:cycle}, let $t_0=T$, the first term $\frac{t_0U(t_0)}{t}\rightarrow 0$ as $t\rightarrow +\infty$, while the second term (the integral) takes the values $G_{i_1, j_1}, \dots, G_{i_K, j_K}$ periodically. 
This leads to the sufficient condition for the convergence of BRU, if a CFP follows a cycle.
\end{proof}

\subsection{Proof of Theorem \ref{thm:lyapnuov}}

\thmlyapunov*

\begin{proof}
Suppose from some initial point $(\mathbf{p}(0),\mathbf{q}(0))$, $\mathbf{G}_\lambda$ enters cycle $C$, since UNE is the unique equilibrium of $\mathbf{G}_0$ and CFP converges on $\mathbf{G}_0$, UNE is also in $\operatorname{Conv}(P(\mathrm{C},\mathbf{G}))$ and $\operatorname{Conv}(P(\mathrm{C},\mathbf{G}))\cap \mathcal{X}(\mathbf{G})\neq \emptyset$.  Since $\dot{U}(\mathbf{p}(t),\mathbf{q}(t))<0$  for all $t>0$, $U$ is strictly decreasing and bounded, and thus converge. Then $\dot{U}$ will converge to 0. Since the only possible point for $\dot{U}$ to be zero is the equilibria in $\operatorname{Conv}(P(\mathrm{C},\mathbf{G}))$, by simple calculus analysis, $(\mathbf{p}(t),\mathbf{q}(t))$ will converge to NE. 

If $G_{\lambda}$ is degenerate, it is possible there exists a local area where every point in the area is a NE of $G_{\lambda}$. For example, if both $(\bfe_1,\bfq)$ and $(\bfe_2,\bfq)$ are NE, then $(\gamma \bfe_1+(1-\gamma)\bfe_2,\bfq)$ ($\gamma\in[0,1]$) forms a convex set of NE. In this case, we can only prove CFP converges to such a set of NE. On the contrast, if $G_{\lambda}$ is non-degenerate, each of its NE must be an isolated point, so that CFP converges to one specific NE.
\end{proof}

\subsection{Experiments on the Shapley game}

Given a strategy pair $(\mathbf{p},\mathbf{q})$, define the following values:
\begin{equation}
    \small{\begin{aligned}
 U(\bfp,\bfq) &= \max_{i\in[m]} (A\bfq)_i +\max_{j\in [n]}(\bfp B)_j \\
V(\mathbf{p}, \mathbf{q}) &=\mathbf{p}^{T}(A+B) \mathbf{q} \\
\text{SE}(\mathbf{p}, \mathbf{q}) &=U(\bfp,\bfq)-V(\bfp,\bfq)\\
\text{ME}(\mathbf{p}, \mathbf{q}) &=\max \left\{\max_{i\in [m]} (A \mathbf{q})_i-\mathbf{p}^{T} A \mathbf{q}, \max_{j \in [n]} (\mathbf{p}^{T} B)_j-\mathbf{p}^{T} B \mathbf{q}\right\}
\end{aligned}}
\end{equation}

First consider how different values related to the beliefs of the agents change in each round. We run the algorithm on the Shapley game for 10,000 rounds with initial point $(1,2)$ on $\mathbf{G}(\lambda)$. The result is shown in Figure \ref{fig:shapley1}, where the line with label ``Sum of epsilons'' is the change trajectory of $SE$, and the line with label ``Max of epsilons'' is the change trajectory of ``ME''. Since we start the sequence from non-symmetric initial condition, from which the FP will not converge, all the other values except $U$ has fluctuations which has no sign of convergence. While $U$, the best response utility function, is quite smooth and converge in the first 2000 rounds.

Given the decomposition of the Shapley game $\mathbf{G}=\mathbf{P}+\mathbf{H}+\mathbf{E}$, as shown in Figure \ref{eqn:shapley}.  Let $\mathbf{G}(\lambda)=\lambda \mathbf{P}+(1-\lambda)\mathbf{H}$. We set the step length of $\lambda$ to be $0.001$ and for each $\lambda$, we run FP for 10,000 rounds with initial condition $(1,2)$ on $\mathbf{G}(\lambda)$ and see how those values varies when the algorithm terminates. First we can find out that with $\lambda$ changing, the best response utility function $U$ still has good smoothness, which implies that it may have some desirable property for us to explore and exploit. Then look at values related to the approximation error: $SE$ and $ME$. When one wants to find an approximation solution for NE, we often use $ME$ to evaluate how good the approximation is. 

Though these two values have fluctuations, we can find out there is a overall trend that the epsilon-related values first increase and then decrease, as $\lambda$ changes from $0$ to $1$. Note that at $\lambda = 0$ and $\lambda = 1$, $\mathbf{G}(\lambda)$ is a harmonic game and a potential game, respectively, on which FP converges. We can reasonably make a conjecture that the overall trend of the approximation error $ME$ is actually an one-pick function. Thus if at any point $\lambda$ in middle of open interval $(0,1)$, $ME=0$, then on all $\lambda$s between 0 and 1, FP will also converge on $\mathbf{G}(\lambda)$. 

As for the fluctuations on the overall trend of epsilon-related values, we can also make such a explanation: When $\lambda$ changes a little bit to $\lambda^\prime$, the FP sequences will enter the same cycle. However, the difference of two $\lambda$s cause the specific strategies along the cycle and the cycle length to be different. Thus when sequences related to $\lambda$ and $\lambda^\prime$ terminate on the same rounds, they will have different approximation error due to the different ``location'' on the cycle. 